\documentclass[journal]{IEEEtran}

\usepackage[T1]{fontenc}
\usepackage{graphicx,latexsym,subfigure,amsmath,epsfig,latexsym,subfigure,amsmath,cite,amssymb}
\usepackage{algorithm,amsthm}
\usepackage{algorithmic}
\usepackage{color}
\usepackage{times}
\usepackage{subeqnarray}
\usepackage{cases}
\usepackage{multirow}
\usepackage{setspace}
\usepackage{url}
    \newtheorem{theorem}{Theorem}

    \newtheorem{corollary}{Corollary}

\begin{document}
\title{Defense Against Advanced Persistent Threats in Dynamic Cloud Storage: A Colonel Blotto Game Approach}
\begin{spacing}{0.5}
\author{Minghui Min\IEEEauthorrefmark{1}, Liang Xiao\IEEEauthorrefmark{1}, Caixia Xie\IEEEauthorrefmark{1}, Mohammad Hajimirsadeghi\IEEEauthorrefmark{2}, Narayan B. Mandayam\IEEEauthorrefmark{2}\\
\IEEEauthorblockA{\IEEEauthorrefmark{1}Department of Communication Engineering, Xiamen University, Xiamen, China. Email: lxiao@xmu.edu.cn}\\
\IEEEauthorblockA{\IEEEauthorrefmark{2}WINLAB, Department of Electrical and Computer Engineering, Rutgers University, Piscataway, NJ. Email: \{mohammad, narayan\}@winlab.rutgers.edu}}
\end{spacing}
\maketitle
\begin{abstract}
Advanced Persistent Threat (APT) attackers apply multiple sophisticated methods to continuously and stealthily steal information from the targeted cloud storage systems and can even induce the storage system to apply a specific defense strategy and attack it accordingly. In this paper, the interactions between an APT attacker and a defender allocating their Central Processing Units (CPUs) over multiple storage devices in a cloud storage system are formulated as a Colonel Blotto game. The Nash equilibria (NEs) of the CPU allocation game are derived for both symmetric and asymmetric CPUs between the APT attacker and the defender to evaluate how the limited CPU resources, the date storage size and the number of storage devices impact the expected data protection level and the utility of the cloud storage system. A CPU allocation scheme based on ``hotbooting'' policy hill-climbing (PHC) that exploits the experiences in similar scenarios to initialize the quality values to accelerate the learning speed is proposed for the defender to achieve the optimal APT defense performance in the dynamic game without being aware of the APT attack model and the data storage model. A hotbooting deep Q-network (DQN)-based CPU allocation scheme further improves the APT detection performance for the case with a large number of CPUs and storage devices. Simulation results show that our proposed reinforcement learning based CPU allocation can improve both the data protection level and the utility of the cloud storage system compared with the Q-learning based CPU allocation against APTs.
\end{abstract}

\section{Introduction}
Cloud storage and cyber systems are vulnerable to Advanced Persistent Threats (APTs), in which an attacker applies multiple sophisticated methods such as the injection of multiple malwares to continuously and stealthily steal data from the targeted cloud storage system \cite{giura2013using},\cite{han2016game} and \cite{wang2011enabling}. APT attacks are difficult to detect and have caused privacy leakage and millions of dollars' loss \cite{vukalovic2015advanced} and \cite{C11}. According to \cite{pone15}, more than 65\% of the organizations in a survey in 2014 have experienced more APT attacks in their IT networks than last year.

The FlipIt game proposed in the seminal work \cite{van2013flipit} formulates the stealthy and continuous APT attacks and designs the scan interval to detect APTs on a given cyber system. The game theoretic study in \cite{zhang2014stealthy} has provided insights to design the optimal scan intervals of a cyber system against APTs. Prospect theory has been applied in \cite{Xiao2016cloudjournal} to investigate the probability distortion of an APT attacker against cloud storage and cumulative prospect theory has been used in \cite{InfocomCPT} to model the frame effect of an APT attacker to choose the attack interval. Most existing APT games ignore the strict resource constraints in the APT defense, such as the limited number of Central Processing Units (CPUs) of a storage defender and an APT attacker \cite{van2013flipit} and \cite{Manshaei2013}. However, a cloud storage system with limited number of CPUs cannot scan all the data stored on the storage devices in a given time slot. To this end, encryptions and authentication techniques are applied to protect data privacy for cloud storage systems. On the other hand, an APT attacker with limited CPU resources cannot install malwares to steal all the data on the cloud storage system in a single time slot either \cite{zhang2015game}.

It is challenging for a cloud storage system to optimize the CPU allocation to scan the storage devices under a large number of CPUs and storage devices without being aware of the APT attack strategy. Therefore, we use the Colonel Blotto game (CBG), a two-player zero-sum game with multiple battlefields to model the competition between an APT attacker and a storage defender, each with a limited total number of CPUs over a given number of storage devices. The player who applies more resources on a battlefield in a Colonel Blotto game wins it, and the overall payoff of a player in the game is proportional to the number of the winning battlefields \cite{roberson2006colonel}. The Colonel Blotto game has been recently applied to design the spectrum allocation of network service providers \cite{hajimirsadeghi2016inter}, the jamming resistance methods for Internet of Things \cite{labib2015colonel} and \cite{namvar2016jamming}.

Our previous work in \cite{CBGPHC} assumes that each storage device has the same amount of data and addresses APT attackers that does not change the attack policy. However, the storage devices usually have different amount of the data with different priority levels, and the data size and their priority level also change over time. By allocating more CPUs to scan the storage devices with more data, a storage defender can achieve a higher data protection level. Therefore, this work extends to a dynamic cloud storage system whose data size changes over time and addresses smart APTs, in which an attacker that can learn the defense strategy first chooses the attack strength to induce the storage system to apply a specific defense strategy and then attacks it accordingly.

By applying time sharing (or division), a defender can use a single CPU to scan multiple storage devices as battlefields to detect APTs in a time slot, and an attacker can use a single CPU to attack multiple devices with a single CPU yielding a roughly continuous CBG. According to \cite{roberson2006colonel}, a pure-strategy Colonel Blotto game rarely achieves Nash equilibria (NEs). Therefore, we focus on the CBG with mixed strategies, in which both players choose their CPU allocation distribution and introduce randomness in their action selection to fool their opponent. The conditions under which the NEs exist in the CPU allocation game are provided to disclose how the number of storage devices, the size of the data stored in each storage device and the total number of CPUs in which the defender observes the impact on the data protection level and the utility of the cloud storage system against APTs.

The CBG-based CPU allocation game provides a framework to understand the strategic behavior of both sides, and the NE strategy relies on the detailed prior knowledge about the APT attack model. In particular, the cloud defender has to know the total number of the attack CPUs and the attack policy over the storage devices, which is challenging to accurately estimate in a dynamic storage system. On the other hand, the repeated interactions between the APT attacker and the defender over multiple time slots can be formulated as a dynamic CPU allocation game, and the defender can choose the security strategy according to the attack history. The APT defense decisions in the dynamic CPU allocation game can be approximately formulated as a Markov decision process (MDP) with finite states, in which the defender observes the state that consists of the previous attack CPU allocation and the current data storage distribution. Therefore, a defender can apply reinforcement learning (RL) techniques such as Q-learning to achieve the optimal CPU allocation over the storage devices to detect APTs in a dynamic game.

The policy hill-climbing (PHC) algorithm as an extension of Q-learning in the mixed-strategy game \cite{bowling2001rational} enables an agent to achieve the optimal strategy without being aware of the underlying system model. For instance, the PHC-based CPU allocation scheme as proposed in our previous work in \cite{CBGPHC} enables the defender to protect the storage devices with limited number of CPUs without being aware of the APT attack model. In this work, a ``hotbooting'' technique as a combination of transfer learning \cite{pan2010survey} and RL exploits experiences in similar scenarios to accelerate the initial learning speed. We propose a ``hotbooting'' PHC-based CPU allocation scheme that chooses the number of the CPUs on each storage device based on the current state and the quality or Q-function that is initialized according to the APT detection experiences to reduce the exploration time at the initial learning stage.


We apply deep Q-network (DQN), a deep reinforcement learning technique recently developed by Google DeepMind in \cite{mnih2015human} and \cite{mnih2013playing} to accelerate the learning speed of the defender for the case with a large number of storage devices and defense CPUs. More specifically, the DQN-based CPU allocation exploits the deep convolutional neural network (CNN) to determine the Q-value for each CPU allocation and thus suppress the state space observed by the cloud storage defender. Simulation results demonstrate that this scheme can improve the data protection level, increase the APT attack cost, and enhance the utility of the cloud storage system against APTs.

The main contributions of this paper are summarized as follows:

(1) We formulate a CBG-based CPU allocation game against APTs with time-variant data size and attack policy. The NEs of the game are provided to disclose the impact of the number of storage devices, the amonut of data stored in each device and the total number of CPUs on the data protection level of the cloud storage system.

(2) A hotbooting PHC-based CPU allocation scheme is developed to achieve the optimal CPU allocation over the storage devices with low computational complexity and improve the data protection level compared with the PHC-based scheme as proposed in \cite{CBGPHC}.

(3) A hotbooting DQN-based CPU allocation scheme is proposed to furture accelerate the learning speed and improve the resistance against APTs.

The rest of this paper is organized as follows: We review the related work in  section \uppercase\expandafter{\romannumeral2}, and present the system model in section \uppercase\expandafter{\romannumeral3}. We formulate a CBG-based CPU allocation game and derive the NEs of the game in Section \uppercase\expandafter{\romannumeral4}. A hotbooting PHC-based CPU allocation scheme and a hotbooting DQN-based scheme are developed in Sections \uppercase\expandafter{\romannumeral5} and \uppercase\expandafter{\romannumeral6}, respectively. Simulation results are provided in Section \uppercase\expandafter{\romannumeral7} and the conclusions of this work are drawn in Section \uppercase\expandafter{\romannumeral8}.

\section{Related work}
The seminal work in \cite{van2013flipit} formulates a stealthy takeover game between an APT attacker and a defender, who compete to control a targeted cloud storage system. The APT scan interval on a single device has been optimized in \cite{bowers2012defending} based on the FlipIt model without considering the constraint of scanning CPUs. The game between an overt defender and a stealthy attacker as investigated in \cite{laszka2013mitigating} provides the best response of the periodic detection strategy against a non-adaptive attacker. The online learning algorithm as developed in\cite{zheng2016reset} achieves the optimal timing of the security updates in the FlipIt game, and reduces the regret of the upper confidence bound compared with the periodic defense strategy. The APT defense game formulated in \cite{zhang2015game} extends the FlipIt game in \cite{van2013flipit} to multi-node systems with limited resources. The game among an APT attacker, a cloud defender and a mobile device as formulated in \cite{pawlick2015flip} combines the APT defense game in \cite{van2013flipit} with the signaling game between the cloud and the mobile device.
The evolutionary game can capture the long term continuous behavior of APTs on cloud storage \cite{IEEEAcessAPT}. The information-trading and APT defense game formulated in \cite{hu2015dynamic} analyzes the joint APT and insider attacks. The subjective view of APT attackers under uncertainty scanning duration was analyzed in \cite{Xiao2016cloudjournal} based on prospect theory.

Colonel Blotto game models the competition between two players each with resource constraints. For example, the Colonel Blotto game with mixed-strategy as formulated in \cite{hajimirsadeghi2016inter} studies the spectrum allocation of network service providers, yielding a fictitious play based allocation approach to compute the equilibrium of the game with discrete spectrum resources. The anti-jamming communication game as developed in \cite{wu2012anti} optimizes the transmit power over multiple channels in cognitive radio networks based on the NE of the CBG. The CBG-based jamming game as formulated in \cite{labib2015colonel} shows that neither the defender nor the attacker can dominate with limited computational resources. The CBG-based jamming game as formulated in \cite{namvar2016jamming} shows how the number of subcarriers impacts the anti-jamming performance of Internet of Things with continuous and asymmetric radio power resources. The CBG-based phishing game as formulated in \cite{chia2011colonel} investigates the dynamics of the detect-and-takedown defense against phishing attacks.

Reinforcement learning techniques have been used to improve network security. For instance, the minimax-Q learning based spectrum allocation as developed in \cite{wang2011anti} increases the spectrum efficiency in cognitive radio networks. The DQN-based anti-jamming communication scheme as designed in \cite{HanDQN2017} applies DQN to choose the transmit channel and node mobility and can increase the signal-to-interference-plus-noise ratio of secondary users against cooperative jamming in cognitive radio networks. The PHC-based CPU allocation scheme as proposed in \cite{CBGPHC} applies PHC to improve the data protection level of the cloud storage system against APTs. Compared with our previous work in \cite{CBGPHC}, this work improves the game model by incorporating the time-variant data storage model. We also apply both the hotbooting technique and DQN to accelerate the learning speed and thus improve the security performance for the case with a large number of storage devices and CPUs against the smart APT attacks in the dynamic cloud storage system.

\section{System Model}
As illustrated in Fig. \ref{system1}, the cloud storage system consists of $D$ storage devices, where device $i$ stores data of size $B_i^{(k)}$ at time $k$, with $1 \leq i \leq D$. Let $\mathbf{B}^{(k)} =  \left[B_i^{(k)}\right]_{1\leq i \leq D}$ be the data size vector of the cloud storage system, and $\widehat{B}^{(k)} = \sum_{i=1}^D {B_i^{(k)}}$ denote the total amount of the data stored in the cloud storage system at time $k$..
\begin{figure}[!htbp]
\begin{center}
\includegraphics[height=1.2 in]{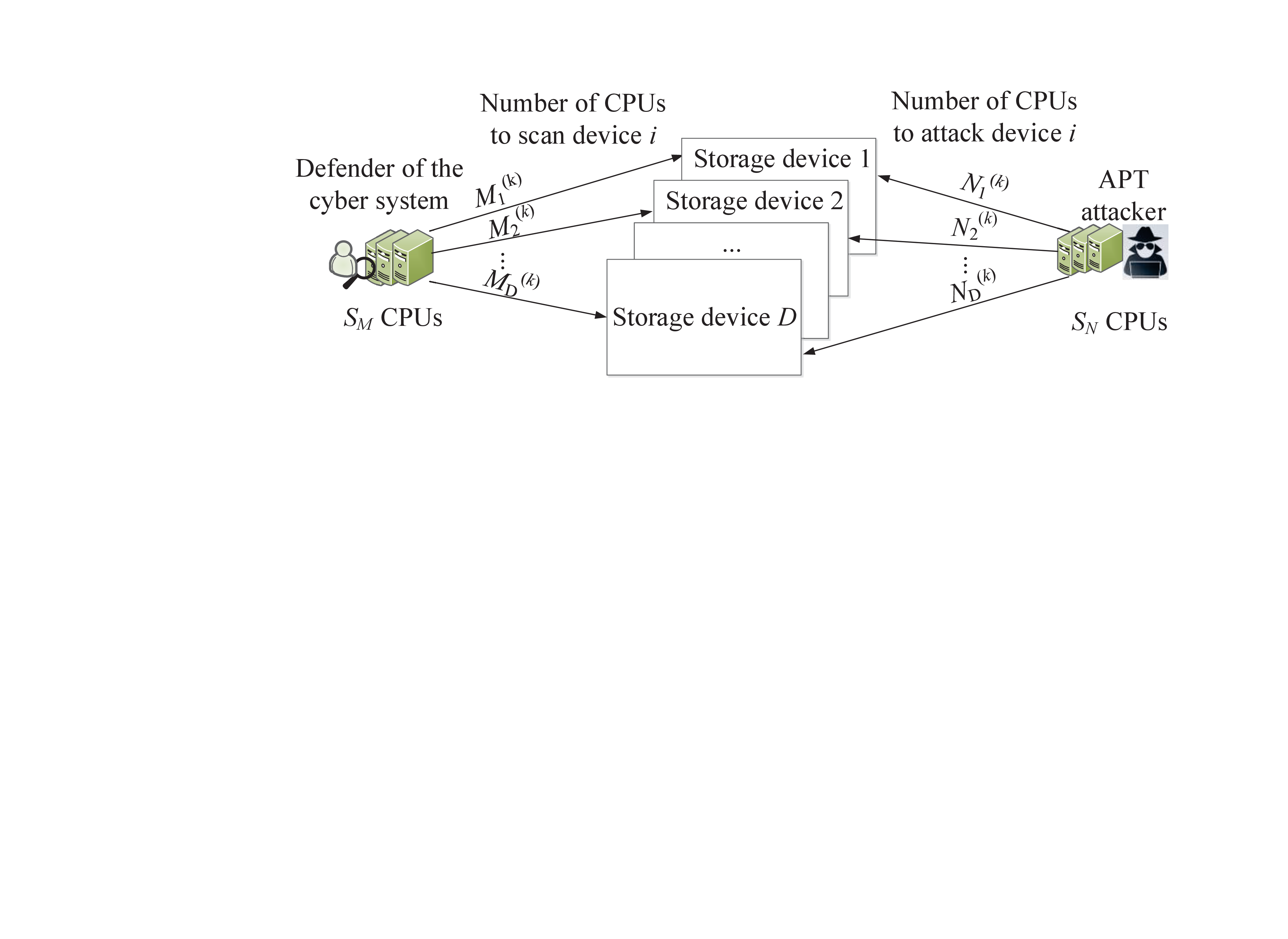}\\
\caption{CPU allocation game, in which a defender with $S_M$ CPUs chooses the CPU allocation strategy to scan the $D$ storage devices in the cloud storage system against an APT attacker with $S_N$ CPUs.}\label{system1}
\end{center}
\end{figure}

In this work, we consider an APT attacker who combines multiple attack methods, tools, and techniques such as \cite{C11} to steal data from the targeted cloud storage system over a long time. The attacker aims to steal more data from the $D$ storage devices with $S_N$ CPUs without being detected. At time $k$, $N_i^{(k)}$ out of the $S_N$ CPUs are used to attack storage device $i$, with $\sum\nolimits_{i = 1}^D {N_i^{(k)}}  \le {S_N}$. The attack CPU allocation at time $k$ is given by $\textbf{N}^{(k)} = \left[N_i^{(k)}\right]_{1\leq i \leq D} \in \triangle_A $, where the attack action set $\triangle_A$ is given by
\begin{align}\label{A_ActionSet}
\triangle_A = \left\{\left[ N_i\right]_{1\leq i \leq D} \bigg| \ {0 \leq N_i \leq  S_N}; \ \sum _{i=1}^D {N_i} \leq S_N \right\}.
\end{align}

A defender uses $S_M$ CPUs to scan the $D$ storage devices in the cloud storage system and aims to detect APTs as early as possible to minimize the total amount of the data stolen by the attacker. At time $k$, $M_i^{(k)}$ out of the $S_M$ CPUs are allocated to scan the device $i$ for APTs, with $\sum\nolimits_{i = 1}^D {M_i^{(k)}} \le {S_M}$. As each time slot is quite short, the storage defender can not scan all the data stored in the $D$ storage devices in a single time slot. The defense CPU allocation vector denoted by $\textbf{M}^{k}$ is defined as $\textbf{M}^{(k)} = \left[M_i^{(k)}\right]_{1\leq i \leq D} \in \triangle_D$, where the defense action set $\triangle_D$ is given by
\begin{align}\label{D_ActionSet}
\triangle_D = \left\{\left[M_i\right]_{1 \leq i \leq D} \bigg|\ {0 \leq M_i \leq  S_M}; \ \sum _{i=1}^D {M_i} \leq S_M \right\}.
\end{align}

If the attacker uses more CPUs than the defender in the APT defense game, the data stored in the storage device are assumed to be at risk. More specifically, the data stored in storage device $i$ are assumed to be safe if the number of the defense CPUs is greater than the number of attack CPUs at that time, i.e., $M_i^{(k)} > N_i^{(k)}$, and the data are at risks if $M_i^{(k)} < N_i^{(k)}$. If $M_i^{(k)} = N_i^{(k)}$, both players have an equal opportunity to control the storage device. Let $\textrm{sgn}(x)$ denote a sign function, with $\textrm{sgn}(x) = 1$ if $x>0$, $\textrm{sgn}(x) = -1$ if $x<0$, and 0 otherwise. Therefore, the data protection level of the cloud storage system at time $k$ denoted by $R^{(k)}$ is defined as the normalized size of the ``safe'' data that are protected by the defender and is given by
\begin{align}\label{R}
R^{(k)} = \frac{1}{\widehat{B}^{(k)}}\sum\limits_{i = 1}^D {B_i^{(k)}}{\textrm{sgn}} \left(M_i - N_i \right).
\end{align}
For ease of reference, our commonly used notations are summarized in Table \ref{table}. The time index $k$ in the superscript is omitted if no confusion occurs.
\begin{table}[!htbp]
\renewcommand{\arraystretch}{1.3}
\caption{SUMMARY OF SYMBOLS AND NOTATIONS}
\label{table_example}
\centering
\begin{tabular}{|c|l|}
\hline
$D$     & Number of storage devices\\
\hline
$ S_{M/N}$  &  Total number of defense/attack CPUs \\
\hline
$\textbf{M}^{(k)}/\textbf{N}^{(k)}$  &Defense/attack CPU allocation vector\\
\hline
$\triangle_{D/A}$ & Action set of the defender/attacker \\
\hline
$u^{(k)}_{D/A}$ & Utility of the defender/attacker at time $k$ \\
\hline
$\widehat{B}^{(k)}$ & Total size of the stored data at time $k$\\
\hline
$\textbf{B}^{(k)}$ & Data size vector of $D$ devices at time $k$\\
\hline
$R^{(k)}$ &  Data protection level at time $k$\\
\hline
\end{tabular}\label{table}
\end{table}

\section{CBG-Based CPU Allocation Game}
Colonel Blotto game is a powerful tool to study the strategic resource allocation of two agents each with limited resources in a competitive environment. Therefore, the interactions between the APT attacker and the defender of the cloud storage system regarding their CPU allocations can be formulated as a Colonel Blotto game with $D$ battlefields. By applying the time sharing (or division) technique, the defender (or attacker) can scan (or attack) multiple storage devices with a single CPU in a time slot, which can be approximately formulated as a continuous CBG. In this game, the defender chooses the defense CPU allocation vector  $\textbf {M}^{(k)} \in \triangle_D $ to scan the $D$ devices at time $k$, while the APT attacker chooses the attack CPU allocation $\textbf {N}^{(k)} \in \triangle_A $.

The utility of the defender (or the attacker) at time $k$ denoted by $u_D^{(k)}$ (or $u_A^{(k)}$) depends on the size of the data stored in the $D$ devices, and the data protection level of each device at the time. In the zero-sum game, by (\ref{R}) the utility of the defender is set as
\begin{align} \label{u}\notag\
&u_D^{(k)}\left( {{\textbf{M}},{\textbf{N}}} \right) = -{u_A^{(k)}}\left( {{\textbf{M}},{\textbf{N}}} \right) \\
&= \sum\limits_{i = 1}^D {B_i^{(k)}}{\textrm{sgn}} \left(M_i - N_i \right).
\end{align}
The CBG-based CPU allocation game rarely has a pure-strategy NE, because the attack CPU allocation $\textbf{N}^{(k)}$ can be chosen according to the defense CPU allocation $\textbf{M}^{(k)}$ and to defeat it for a higher utility $u_A^{(k)}$. Therefore, we study the CPU allocation game with mixed strategies, in which each player randomizes the CPU allocation strategies to fool the opponent.

In the mixed-strategy CPU allocation game, the defense strategy at time $k$ denoted by $x_{i,j}^{(k)}$ is the probability that the defender allocates $j$ CPUs to scan device $i$, i.e., $x_{i,j}^{(k)}= \mathrm{Pr}\left( M_i^{(k)} = j \right)$. Let $p_{m,j} \in \left[0,1\right]$ be the $m$-th highest feasible value of $x_{i,j}^{(k)}$. The mixed-strategy defense action set denoted by $\textrm{P}_{M}$ is given by
\begin{align} \nonumber
\textrm{P}_{M}= \bigg\{[p_{m,j}]_{ 1\leq m \leq D,\ 0 \leq j \leq S_M} \bigg| p_{m,j} \geq 0, \ \forall j; \\ \sum _{j = 0}^{S_M} p_{m,j} =1, \ \forall m  \bigg\}.
\end{align}
The defense mixed strategy vector denoted by $\textbf{x}^{(k)}$ is given by
\begin{align}
\textbf{x}^{(k)} = \left[x_{i,j}^{(k)}\right]_{1\leq i \leq D,\  0 \leq j \leq S_M} \in \textrm{P}_{M}.
\end{align}

Similarly, let $y_{i,j}^{(k)}$ denote the probability that $N_i^{(k)}$ CPUs are used to attack device $i$, i.e., $y_{i,j}^{(k)} = \mathrm{Pr}\left( N_i^{(k)} = j \right)$, and $q_{m,j} \in \left[0,1\right]$ be the $m$-th highest feasible value of $y_{i,j}^{(k)}$. The action set of the attacker in the mixed-strategy game denoted by $\textrm{P}_{N}$ is given by
\begin{align} \nonumber
\textrm{P}_{N}= \bigg\{[q_{m,j}]_{ 1\leq m \leq D,0 \leq j \leq S_N } \bigg| \ q_{m, j} \geq 0, \ \forall j; \\ \sum _{j = 0}^{S_N} q_{m, j} =1,  \ \forall m   \bigg\}.
\end{align}
 The attacker chooses the CPU allocation strategy in this game denoted by $\textbf{y}^{(k)}$ with
\begin{align}
\textbf{y}^{(k)} = \left[y_{i,j}^{(k)}\right]_{1\leq i\leq D, \ 0 \leq j \leq S_N} \in \textrm{P}_{N}.
\end{align}

The expected utility of the defender (or the attacker) averaged over all the feasible defense (or attack ) strategies is denoted by $U_D^{(k)}$ (or $U_A^{(k)}$) and given by (\ref {u}) as
\begin{align} \label {utility of attacker}\notag\
&{U_D^{(k)}}\left(\textbf{x},\textbf{y}\right)= - {U_A^{(k)}}\left(\textbf{x},\textbf{y}\right)\\
&=  E_{\substack{\textbf{M} \sim \textbf{x}\\
  \textbf{N} \sim \textbf{y}}}\left(\sum\limits_{i = 1}^D {B_i^{(k)}}{\textrm{sgn}} \left( {M_i - N_i} \right)\right).
\end{align}

The NE of the CBG-based CPU allocation game with mixed strategies denoted by $\left(\textbf{x}^{*}, \textbf{y}^{*}\right)$ provides the best-response policy, i.e., no player can increase his or her utility by unilaterally changing from the NE strategy. For example, if the defender chooses the CPU allocation strategy $\textbf{x}^{*}$, the APT attacker cannot do better than selecting $\textbf{y}^{*}$ to attack the $D$ storage devices. By definition, we have
\begin{align} \label {UD}
{U_D}\left(\textbf{x}^{*},\textbf{y}^{*}\right) & \ge {U_D}\left(\textbf{x},\textbf{y}^{*}\right), \ \ \forall \ \textbf{x} \in \textrm{P}_{M}\\ \label {UA}
{U_A}\left(\textbf{x}^{*},\textbf{y}^{*}\right) & \ge {U_A}\left(\textbf{x}^{*},\textbf{y}\right), \  \ \forall \ \textbf{y} \in \textrm{P}_{N}.
\end{align}

We first consider a CBG-based CPU allocation game  $\mathbb{G}_1$ with symmetric CPU resources, $S_M = S_N$, i.e., the defender and the attacker have the same amount of computational resources. Let $\mathbf{1}_{m \times n}$ (or $\mathbf{0}_{m \times n}$) be an all-1 (or 0) $m \times n$ matrix, $\left\lfloor \ \right\rfloor$ be the lower floor function, and the normalized defense CPUs $\beta = 2S_M/ \widehat{B}$.

\begin{theorem}\label{theorem1}
If $S_M = S_N$ and $B_i < \sum{_{1 \leq h \ne i \leq D}{B_h}}$, the CPU allocation game $\mathbb{G}_1$ has a NE $\left(\textbf{x}^*, \textbf{x}^*\right)$ given by
\begin{equation}\label{theorem1-equ}
\textbf{x}^* =
\begin{bmatrix}
\frac{1}{\left\lfloor \beta B_1\right\rfloor +1} \mathbf{1}_{1 \times \left({\lfloor{\beta B_1}\rfloor}+1\right)} &\mathbf{0}_{1 \times {\left({S_M - \lfloor{\beta B_1}\rfloor}\right)}}\\
\frac{1}{\lfloor \beta B_2\rfloor +1} \mathbf{1}_{1 \times \left({\lfloor{\beta B_2}\rfloor}+1\right)} &\mathbf{0}_{1 \times {\left({S_M - \lfloor{\beta B_2}\rfloor}\right)}}\\
\vdots & \vdots \\
\frac{1}{\lfloor \beta B_D \rfloor +1} \mathbf{1}_{1 \times \left({\lfloor{\beta B_D}\rfloor}+1\right)} &\mathbf{0}_{1 \times {\left({S_M - \lfloor{ \beta B_D}\rfloor}\right)}}\\
\end{bmatrix}.
\end{equation}
\end{theorem}

\begin{proof}
The CPU allocation game $\mathbb{G}_1$ can be formulated as a CBG with symmetric players on $D$ battlefields. The resource budget of the defender is $S_M$, the value of the $i$-th battlefield is $B_i$, and the total value of $D$ battlefields is $\widehat{B}$.
Let $\mathcal{U}(m, n)$ denote the uniform distribution between $m$ and $n$. By Proposition 1 in \cite{Thomas2017}, the mixed-strategy CBG game has an NE given by $(\textbf{x}^*, \textbf{x}^*)$, where $\textbf{x}^*$ is the probability distribution of $\textbf{M}$, and each vector coordinate $M_i$ is uniform distribution between 0 and $2S_M B_i/ \widehat{B}$.
Therefore, the CPU allocation of the $i$-th device $M_i^*$ is uniformly distributed on $[0, 2S_M B_i/ \widehat{B}]$, i.e.,
\begin{align} \label {7}
M_i^* \sim \mathcal{U}\left(\left\{0,1,2,..., \bigg\lfloor \frac{2S_M B_i}{\widehat{B}}\bigg\rfloor\right\}\right).
\end{align}
Thus this game has an NE given by $\left(\textbf{x}^*, \textbf{x}^*\right)$, where $\forall \  0 \leq j \leq  \lfloor 2S_M B_i\widehat{B}\rfloor$, $0\leq i \leq D$, each element of $\textbf{x}$ is given by
\begin{align} \label {Mi}
x_{i,j}^* = \mathrm{Pr}\left( M_i^* = j \right) = & \frac{1}{\lfloor 2S_MB_i/ \widehat{B}\rfloor +1},
\end{align}
which results in (\ref {theorem1-equ}).
\end{proof}
\begin{corollary}
At the NE of the symmetric CPU allocation game $\mathbb{G}_1$, the expected data protection level is zero and the utility of the defender is zero.
\end{corollary}
\begin{proof}
By (\ref {R}) and (\ref {theorem1-equ}), the data protection level over all the realizations of the mixed-strategy NE $(\textbf{x}^*,\textbf{x}^*)$ is given by
\begin{align} \label {zhengming1}
&E_{\textbf{x}^*}\left(R\right) = E_{\textbf{x}^*}\left(\frac{1}{\widehat{B}}\sum \limits_{i=1}^D {B_i}\textrm{sgn}\left(M_i^* - N_i^*\right)\right)\\ \label {zero}
&= \frac{1}{\widehat{B}}\sum \limits_{i=1}^D {B_i} \bigg(\mathrm{Pr}\left(N_i^* < M_i^*\right) - \mathrm{Pr}\left(N_i^* > M_i^*\right) \bigg)=0.
\end{align}
Similarly, by (\ref{u}) and (\ref{utility of attacker}) , we have $U_D = U_A = 0$.
\end{proof}
\textbf{Remark}:
If the APT attacker and the defender have the same number of CPUs and no storage device dominates in the game (i.e., $B_i < \sum{_{1 \leq h \ne i \leq D}{B_h}}$, $\forall i$), both players choose a number from $\{0,1,...,\lfloor 2S_M B_i/ \widehat{B}\rfloor\}$ to attack or scan storage device $i$ with probability $1 / \left({\lfloor 2S_M B_i/ \widehat{B}\rfloor +1}\right)$ by (\ref {Mi}). The data protection level $R$ by definition ranges between $-1$ and $1$. Therefore, the game makes a tie, yielding zero expected data protection level and zero utility of the defender.

We next consider a CBG-based CPU allocation game with asymmetric players denoted by $\mathbb{G}_2$, in which the attacker and the defender have different number of CPUs and compete over $D$ storage devices with an equal data size, i.e., $ B_i = B$, $\forall i$.
\begin{theorem}\label {theorem2}
If $ 2/D \le S_N/S_M \le 1$, $D \geq 3$ and $ B_i = B$, $\forall \ 1 \leq i \leq D$, the NE of the CPU allocation game $\mathbb{G}_2$ $\left(\textbf{x}^*, \textbf{y}^*\right)$ is given by
\begin{align}{\label {asymmetricx}}
\textbf{x}^* & =\left[\mathbf{0}_{D \times 1} \ \ \ \frac{1}{\left\lfloor \frac{2S_M}{D}\right\rfloor} \mathbf{1}_{D \times \left\lfloor \frac{2S_M}{D}\right\rfloor} \ \ \ \mathbf{0}_{D \times \left({S_M-\left\lfloor \frac{2S_M}{D}\right\rfloor}\right)} \right] \\{\label {asymmetricy}}
\textbf{y}^*& = \bigg[\left(1-\frac{S_N}{S_M}\right)\mathbf{1}_{D \times 1} \ \ \ \left(\frac{S_N}{S_M \left\lfloor \frac{2S_M}{D}\right\rfloor}\right) \mathbf{1}_{D \times \left\lfloor \frac{2S_M}{D}\right\rfloor} \ \ \ \cr
&\ \ \ \ \ \ \ \ \ \ \ \ \ \ \  \ \ \ \ \ \ \ \ \ \ \ \  \ \  \ \ \  \  \qquad \mathbf{0}_{D \times \left({S_M-\left\lfloor \frac{2S_M}{D}\right\rfloor}\right)}\bigg].
\end{align}
\end{theorem}

\begin{proof}
The CPU allocation game $\mathbb{G}_2$ can be formulated as a CBG with asymmetric players on $D$ battlefields, where the defender (or attacker) chooses the probability density functions $\textbf{x}$ (or $\textbf{y}$) according to $S_M$ (or $S_N$) resource budget, and the resources allocated to the $i$-th battlefield is $M_i$ (or $N_i$).

By Theorem 2 in \cite{roberson2006colonel}, the unique Nash equilibrium for the defender and the attacker with $2/D \leq S_N/S_M \leq 1$ is given by
\begin{align} \label{Ep1}
\textbf{x}(M_i^*) &\sim \mathcal{U}\left( {\left[ {0,\frac{2{S_M}}{D}} \right]} \right) \\
\textbf{y}(N_i^*) &\sim \left(1 - \frac{S_N}{S_M}\right)\delta (N_i^*) + \frac{S_N}{S_M}\mathcal{U}\left( {\left[ {0,\frac{2{S_M}}{D}} \right]} \right).
\end{align}

Therefore, the CPU allocation of the $i$-th storage device $M_i^*$ on NE is uniformly distributed on $[0, 2{S_M}/{D}]$, i.e.,
\begin{align}\label {continuousx}
M_i^* \sim \mathcal{U}\left(\left\{0,1,2,..., \bigg\lfloor \frac{2S_M}{D}\bigg\rfloor\right\}\right).
\end{align}
Thus, the NE strategy of the CPU allocation game $\mathbb{G}_2$ is given by
\begin{align}
x_{i,j}^* &= \mathrm{Pr}\left( M_i^* = j \right) = \frac{1}{\left\lfloor \frac{2S_M}{D} \right\rfloor}, \ \forall \  1 \leq j \leq \left\lfloor \frac{2S_M}{D} \right\rfloor.
\end{align}
Thus, we have ({\ref {asymmetricx}}).

Similarly, we have
\begin{align}\label {continuousy}
N_i^*\sim \left(1 - \frac{S_N}{S_M}\right)\delta (N_i^*) + \frac{S_N}{S_M}\mathcal{U}\left(\left\{0,1,2,..., \bigg\lfloor \frac{2S_M}{D}\bigg\rfloor\right\}\right),
\end{align}
and thus
\begin{align}
y_{i,j}^* &= \mathrm{Pr}\left( N_i^* = j \right) =
\begin {cases}
1-\frac{S_N}{S_M}, &  \textrm{if} \   j = 0\\
\frac{S_N}{S_M \left\lfloor \frac{2S_M}{D}\right\rfloor},& \textrm{if} \  1 \leq j \leq \left\lfloor \frac{2S_M}{D} \right\rfloor.
\end{cases}
\end{align}
Thus, we have ({\ref {asymmetricy}}).
\end{proof}
\begin{corollary}
At the NE of the CPU allocation game $\mathbb{G}_2$, the expected data protection level is $1- {S_N}/{S_M}$ and
\begin{align}  \label{T2U}
U_D = -U_A = \left(1- \frac{S_N}{S_M}\right)\widehat{B}.
\end{align}
\end{corollary}
\begin{proof}
According to (\ref {R}), (\ref {zhengming1}), (\ref {asymmetricx}) and (\ref {asymmetricy}), as $B_i = B$, we have
\begin{align} \nonumber
&E_{\textbf{x}^*, \textbf{y}^*}\left(R\right)  = \frac{1}{\widehat{B}} \sum \limits_{i=1}^D {B_i} \left(\mathrm{Pr}\left(N_i^* < M_i^*\right) - \mathrm{Pr}\left(N_i^* > M_i^*\right) \right)\\  \nonumber
& = \frac{1}{\widehat{B}}\sum \limits_{i=1}^D {B_i} \bigg(\mathrm{Pr}\left(M_i^* > 0 \right) + \mathrm{Pr}\left(N_i^* < M_i^*\bigg|{N_i^*\neq0}\right) \cr
& \ \ \ \ -\mathrm{Pr}\left(N_i^* > M_i^*\bigg|{N_i^*\neq0}\right)\bigg)\\ \label {12}
&= \frac{1}{\widehat{B}}\sum \limits_{i=1}^D {B_i}\left(1- \frac{S_N}{S_M}\right) = 1-\frac{S_N}{S_M}.
\end{align}

Similarly, by (\ref{u}) and (\ref{utility of attacker}), we have (\ref {T2U}).
\end{proof}
\textbf{Remark}: The defender has to have more CPU resources than APT attackers, otherwise the cloud storage system is unlikely to protect the data privacy. Therefore, a subset of the storage devices are safe from the attacker who has to match the defender on the other storage devices. In this case, the defender wins the game, and the utility increases with the total data size. The expected data protection level increases with the resource advantage of the defender over the attacker, i.e., $S_M /S_N$.

The APT defense performance of the CPU allocation game $\mathbb{G}_2$ at the NE is presented in Fig. \ref {T2}, in which the $D$ storage devices are threatened by an APT attacker with $150$ attack CPUs. If the defender uses 1200 CPUs instead of 600 CPUs to scan the 20 devices, the data protection level increases about 10.5\% to 93\% and the utility of the defender increases by 18.75\%, The data protection level of the cloud storage system protected by 1200 CPUs slightly decreases by 2.8\%, if the number of the storage devices $D$ changes from 20 to 80. The APT defense performance of the CBG game at the NE provides the optimal defense performance with known APT attack model and defense model, and can be used as a guideline to design the CPU allocation scheme.

\begin{figure}[!t]
\begin{center}
\includegraphics[height=2.2 in]{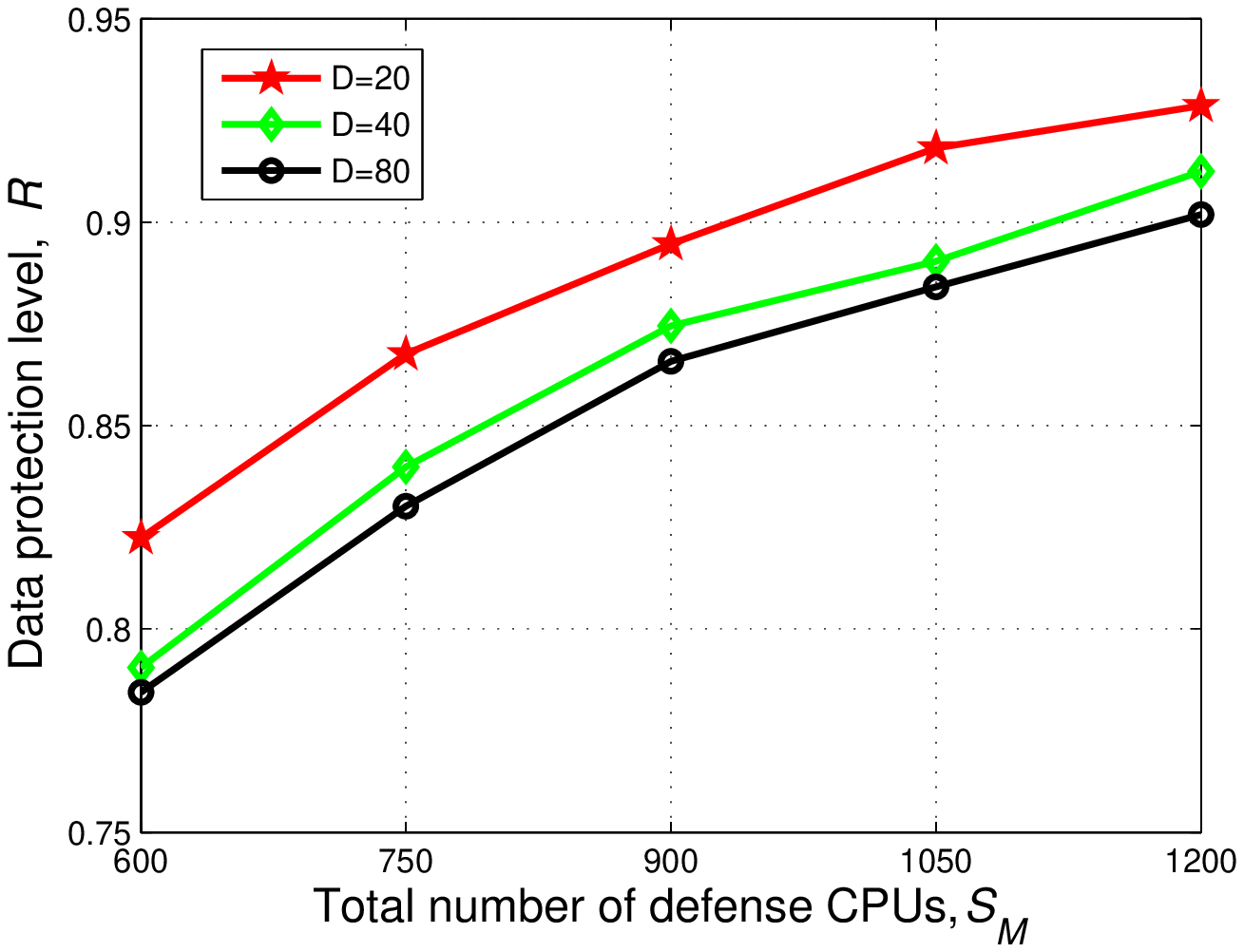}\\
{\footnotesize (a) Data protection level} \\
\includegraphics[height=2.2 in]{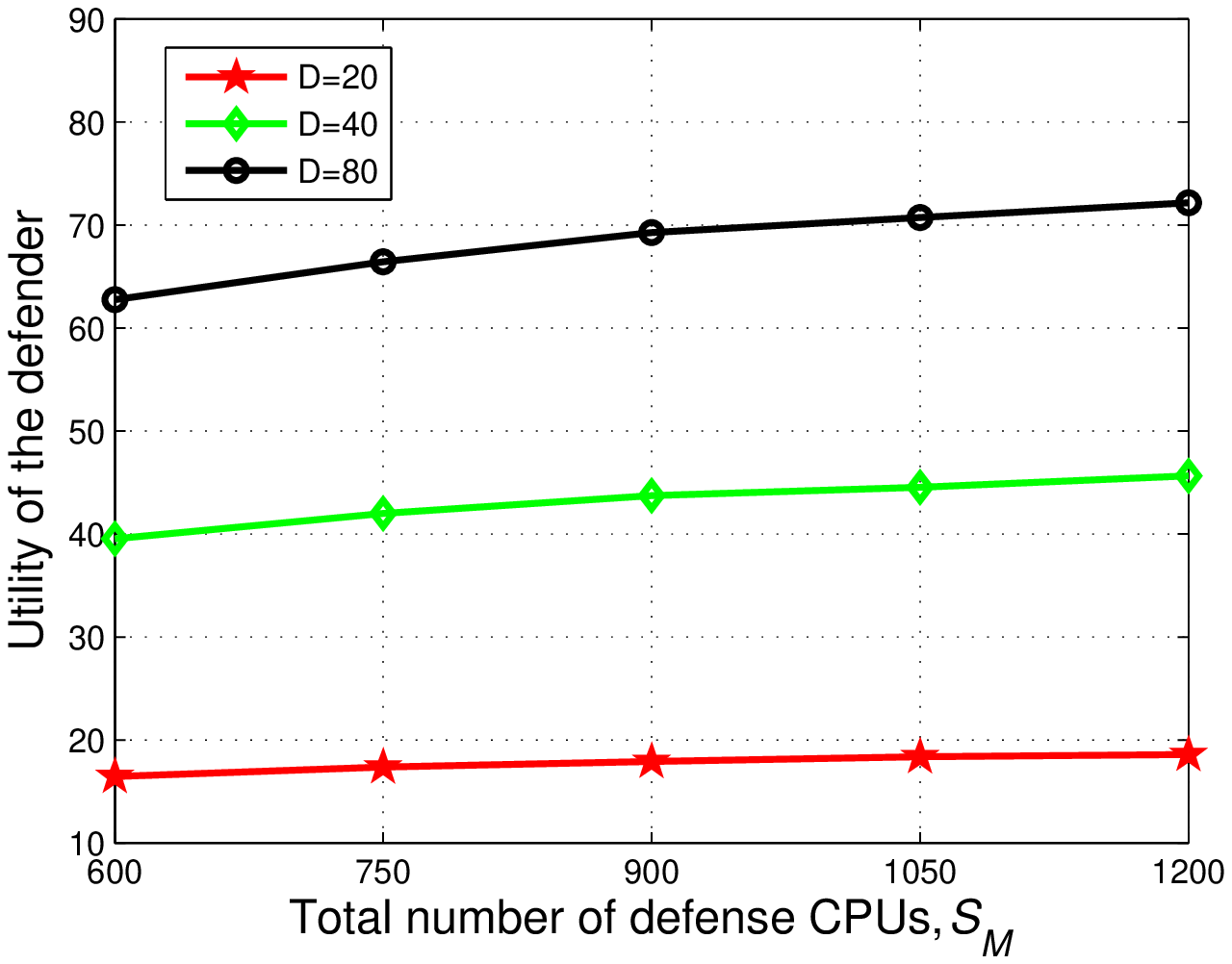}\\
{\footnotesize (b) Utility of the defender} \\
\caption{APT defense performance of the CBG-based CPU allocation game $\mathbb{G}_2$ at the NE with $D$ storage devices and $S_M$ defense CPUs against an APT attacker with 150 CPUs.}\label {T2}
\end{center}
\end{figure}

\section{Hotbooting PHC-based CPU Allocation}
\begin{figure}[!t]
\begin{center}
\includegraphics[height=2.5 in]{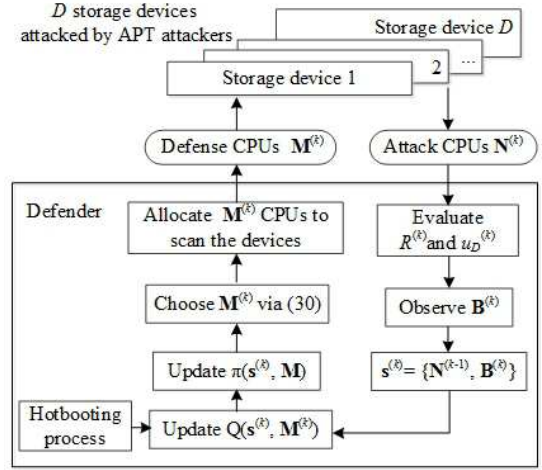}\\
\caption{Illustration of the hotbooting PHC-based defense CPU allocation.}\label{PHC-based system}
\end{center}
\end{figure}
As a defender is usually unaware of the attack policy, we propose a hotbooting PHC-based CPU allocation scheme to scan $D$ storage devices in the dynamic APT detection game, as illustrated in Fig. \ref{PHC-based system}. At each time slot, the defender of the cloud storage system observes the amount of the data stored in each storage device that is quantized into $L$ levels, $B_i^{(k)}\in\{l/L\}_{0\leq l \leq L}$. In addition, the defender also evaluates the compromised storage devices that are found to be attacked by APTs in the last time slot, and uses them to estimate the last attack CPU allocation ${\textbf{N}^{(k-1)}}$. The defense CPU allocation is chosen according to the current state denoted by ${\textbf{s}^{(k)}}$, which consists of the current data sizes and the previous attack CPU allocation, i.e., ${\textbf{s}^{(k)}} =\{ {\textbf{N}^{(k-1)}}, \textbf{B}^{(k)}\}$. The resulting defense strategy ${\textbf{M}^{(k)}} \in \triangle_D $, where $\triangle_D$ is the defense action set given by (\ref{D_ActionSet}).

\begin{algorithm}[h]
\centering
\caption{CPU allocation with hotbooting PHC}\label{hotbooting}
\begin{algorithmic}[1]
\STATE Hotbooting defense process in Algorithm 2
\STATE Initialize $\alpha$, $\gamma$, $\delta$, $\textbf{N}^{(0)}$ and $\textbf{B}^{(1)}$
\STATE Set $\textbf{Q} = \overline{\textbf{Q}}$, $\pi = \overline{\pi}$
\FOR {$k = 1, 2, 3,...$}
\STATE Observe the current data size $\textbf{B}^{(k)}$
\STATE ${\textbf{s}^{(k)}} =\{ {\textbf{N}^{(k-1)}}, \textbf{B}^{(k)}\}$
\STATE Choose $\textbf{M}^{(k)}\in \triangle_D$ with $\pi({\textbf{s}^{(k)}},\textbf{M})$ via (\ref{pi})
\FOR {$i = 1, 2,...D$}
\STATE Allocate $M_i^{(k)}$ CPUs to scan storage device $i$
\ENDFOR
\STATE Observe the compromised storage devices and estimate $\textbf{N}^{(k)}$
\STATE Obtain $u_D^{(k)}$ via (\ref{u})
\STATE Update $Q({\textbf{s}}^{(k)}, \textbf{M}^{(k)})$ via (\ref{Q})
\STATE Update $V({\textbf{s}}^{(k)})$ via (\ref{V})
\STATE Update $\pi({\textbf{s}}^{(k)}, \textbf{M})$ via (\ref{probability})
\ENDFOR
\end{algorithmic}
\end{algorithm}
The Q-function for each action-state pair denoted by $Q\left(\textbf{s},\textbf{M}\right)$ is the expected discounted long-term reward of the defender, and is updated in each time slot according to the iterative Bellman equation as follows:
\begin{align}\nonumber
Q\left({\textbf{s}^{(k)}},\textbf{M}^{(k)}\right) &\leftarrow \left(1 - \alpha \right)Q\left({\textbf{s}^{(k)}},\textbf{M}^{(k)}\right) \\
&+ \alpha \left(u_D^{(k)} + \gamma V\left({\textbf{s}'}\right)\right) \label{Q},
\end{align}
where the learning rate $\alpha \in \left({0,1} \right]$ is the weight of the current experience, the discount factor $\gamma \in \left[ {0,1} \right]$ indicates the uncertainty of the defender on the future reward, $\textbf{s}'$ is the next state if the defender uses $\textbf{M}^{(k)}$ at state $\textbf{s}^{(k)}$, and the value function  $V\left({\textbf{s}}\right)$ maximizes $Q\left(\textbf{s},\textbf{M}\right)$ over the action set given by
\begin{align}
V\left({\textbf{s}^{(k)}}\right) &=\mathop {\max }\limits_{\textbf{M}' \in \triangle_D} Q\left({\textbf{s}^{(k)}},\textbf{M}'\right).\label{V}
\end{align}

The mixed-strategy table of the defender denoted by $\pi \left({\textbf{s}},\textbf{M}\right)$ provides the distribution of the number of CPUs $\textbf{M}$ over the $D$ storage devices under state $\textbf{s}$ and is updated via
\begin{align} \label{probability} \nonumber
&\pi({\textbf{s}^{(k)}},\textbf{M})\leftarrow \pi({\textbf{s}^{(k)}},\textbf{M})\\
&+\begin{cases} \delta, & \text{if $\textbf{M}= \mathop{\arg\max}\limits_{\textbf{M}' \in \triangle_D} Q\left({\textbf{s}^{(k)}},\textbf{M}'\right)$} \\ \frac{\delta}{1-|\triangle_D |}, &\textrm{o.w}.
\end{cases}
\end{align}
In this way, the probability of the action that maximizes the Q-function increases by $\delta$, with $0 < \delta \leq 1$, and the probability of other actions decrease by {${\delta} /\left(|\triangle_D |-1\right)$. The defender then selects the number of CPUs $\textbf{M}^{(k)} \in \triangle_D$ according to the mixed strategy $\pi \left({\textbf{s}^{(k)}},\textbf{M}\right)$, i.e.,
\begin{align}
\Pr \left({\textbf{M}^{(k)}} = {\widehat{\textbf{M}}}\right)= \pi \left({\textbf{s}^{(k)}},\widehat{\textbf{M}}\right),\  \forall \ \widehat{\textbf{M}} \in \triangle_D. \label {pi}
\end{align}
\begin{algorithm}[h]
\centering
\caption{Hotbooting defense process}\label{hotbooting}
\begin{algorithmic}[1]
\STATE Initialize $\xi$, $K$, $\alpha$, $\gamma $, $\delta$, $\textbf{N}^{(0)}$ and $\textbf{B}^{(1)}$
\STATE Set $\textbf{Q} = {\textbf{0}}_{(|\triangle_A| \times L^D) \times |\triangle_D|}$, $\textbf{V} = {\textbf{0}}_{(|\triangle_A|\times L^D) \times 1}$, $\pi =\frac{\textbf{1}}{|\triangle_D |}$\\
\FOR {$i = 1, 2, 3,..., \xi$}
\STATE Emulate a similar CPU allocation scenario for the defender to scan storage devices
\FOR {$k = 1,2,...,K$}
\STATE Observe the current data size $\textbf{B}^{(k)}$
\STATE ${\textbf{s}^{(k)}} =\{ {\textbf{N}^{(k-1)}}, \textbf{B}^{(k)}\}$
\STATE Choose $\textbf{M}^{(k)}\in \triangle_D$ via (\ref{pi})
\FOR {$j = 1, 2,...D$}
\STATE Allocate $M_j^{(k)}$ CPUs to scan storage device $j$
\ENDFOR
\STATE Observe the compromised storage devices and estimate $\textbf{N}^{(k)}$
\STATE Obtain $u_D^{(k)}$ via (\ref{u})
\STATE Update $\textbf{Q}$ and $\pi$ via (\ref{Q})-(\ref{probability})
\ENDFOR
\ENDFOR
\STATE Output $\overline{\textbf{Q}} \leftarrow \textbf{Q}$, $\overline{\pi} \leftarrow \pi$
\end{algorithmic}
\end{algorithm}
We apply the hotbooting technique to initialize both the Q-value and the strategy table $\pi$ with the CPU allocation experiences in similar environments. The hotbooting PHC-based CPU allocation saves random explorations at the beginning stage of the dynamic game and thus accelerates the learning speed. As shown in Algorithm 2, $\xi$ CPU allocation experiences are performed before the game. Each experiment lasts $K$ time slots, in which the defender chooses the number of CPUs to scan the $D$ storage devices according to the mixed-strategy table $\pi \left({\textbf{s}^{(k)}},\textbf{M}\right)$. The defender observes the attack CPU distribution and evaluates the utility $u_D^{(k)}$. Both the Q-function and $\pi$ are updated via (\ref{Q})-(\ref{probability}) in each time in the experiences.

The Q-values as the output of the hotbootng process  based on the $\xi$ experiences denoted by $\overline{\textbf{Q}}$ is used to initialize the Q-values in Algorithm 1. Similarly, the mixed-strategy table as the output of Algorithm 2 based on the $\xi$ experiences denoted by $\overline{\pi}$ is used to initialize $\pi$ in Algorithm 1. The learning time of Algorithm 1 increases with the dimension of the action-state space $|\triangle_D |\times |\triangle_A|\times L^D$, which increases with the number of storage devices in the cloud storage system and the number of CPUs, yielding serious performance degradation.
\begin{figure*}[!t]
\begin{center}
\includegraphics[height= 3.6 in]{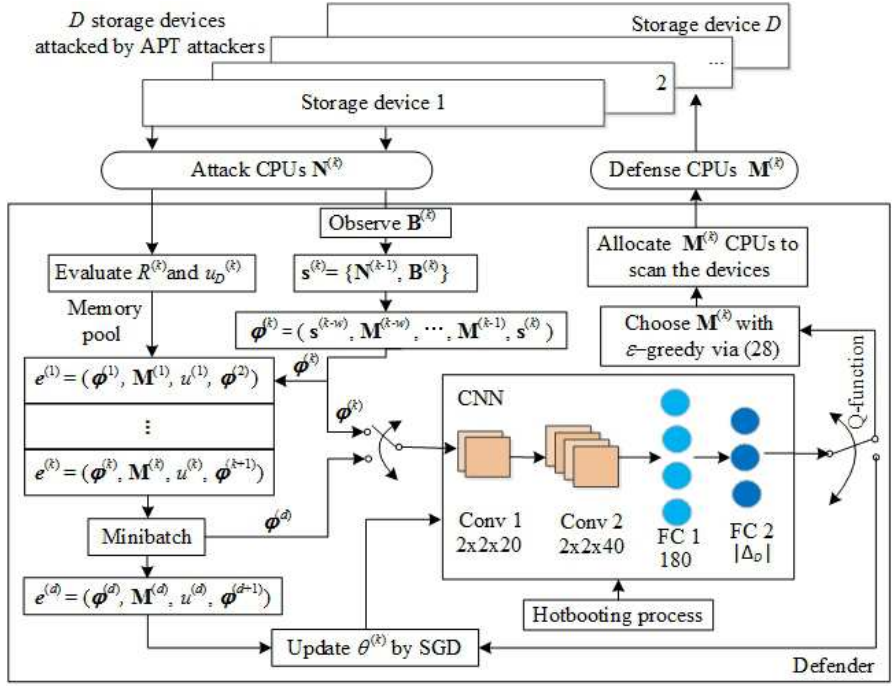} \\
\caption{Hotbooting DQN-based defense CPU allocation.}\label{DQN-based system}
\end{center}
\end{figure*}

\section{Hotbooting DQN-Based CPU Allocation}
In this section, we propose a hotbooting DQN-based CPU allocation scheme to improve the APT defense performance of the cloud storage system. This scheme applies deep convolutional neural network, a deep reinforcement learning technique, to compress the action-state space and thus accelerate the learning process. As illustrated in Fig. \ref{DQN-based system}, the deep convolution neural network is a nonlinear approximator of the Q-value for each action. The CNN architecture allows a compact storage of the learned information between similar states \cite{kaelbling1996reinforcement}.

The DQN-based CPU allocation as summarized in Algorithm 3 extends the system state $\textbf{s}^{(k)}$ as in Algorithm 1 to the experience sequence at time $k$ denoted by $\boldsymbol{\varphi}^{(k)}$ to accelerate the learning speed and improve the APT resistance. More specifically, the experience sequence consists of the current system state $\textbf{s}^{(k)}$ and the previous $W$ system state-action pairs, i.e.,  $\boldsymbol{\varphi}^{(k)} = \left(\textbf{s}^{(k-W)}, \textbf{M}^{(k-W)}, ...,\textbf{s}^{(k-1)},\textbf{M}^{(k-1)}, \textbf{s}^{(k)} \right)$.
\renewcommand{\baselinestretch}{1.0}
\begin{algorithm}[h]
\centering
\caption{Hotbooting DQN-based CPU allocation}\label{hotbooting}
\begin{algorithmic}[1]
\STATE Initialize $\textbf{N}^{(0)}$, $\textbf{B}^{(1)}$, $W$ and $H$
\STATE Set $\theta= \overline{\theta}$, $\mathcal{D} = \emptyset$
\FOR {$k = 1, 2, 3,...$}
\STATE Observe the current data size $\textbf{B}^{(k)}$
\STATE ${\textbf{s}^{(k)}} =\{ {\textbf{N}^{(k-1)}}, \textbf{B}^{(k)}\}$
\IF {$k\leq W$}
    \STATE Choose $\textbf{M}^{(k)} \in \triangle_D$ at random
\ELSE
\STATE $\boldsymbol{\varphi}^{(k)} = \left(\textbf{s}^{(k-W)}, \textbf{M}^{(k-W)}, ...,\textbf{s}^{(k-1)},\textbf{M}^{(k-1)}, \textbf{s}^{(k)} \right)$
\STATE Set $\boldsymbol{\varphi}^{(k)}$ as the input of the CNN
\STATE Observe the output of the CNN to obtain $Q\left(\boldsymbol{\varphi}^{(k)}, \textbf{M}\right)$
\STATE Choose $\textbf{M}^{(k)} \in \triangle_D$ via (\ref{greedy})
\ENDIF
\FOR {$i = 1, 2,...D$}
\STATE Allocate $M_i^{(k)}$ CPUs to scan storage device $i$
\ENDFOR
\STATE Observe the compromised storage devices and estimate $\textbf{N}^{(k)}$
\STATE Obtain $u_D^{(k)}$ via (\ref{u})
\STATE Observe $\boldsymbol{\varphi}^{(k+1)}$
\STATE$\mathcal{D} \leftarrow \mathcal{D} \bigcup \left(\boldsymbol{\varphi}^{(k)}, \textbf{M}^{(k)}, u_D^{(k)}, \boldsymbol{\varphi}^{(k+1)} \right)$
\FOR {$d = 1, 2, 3,...,H$}
\STATE Select $\left(\boldsymbol{\varphi}^{(d)}, \textbf{M}^{(d)}, u_D^{(d)}, \boldsymbol{\varphi}^{(d+1)} \right) \in \mathcal{D}$ at random
\STATE Calculate $G$ via (\ref{R-target})
\ENDFOR
\STATE Update $\theta^{(k)}$ via (\ref{updateweight})
\ENDFOR
\end{algorithmic}
\end{algorithm}
The experience sequence $\boldsymbol{\varphi}^{(k)}$ is reshaped into a $5\times5$ matrix and then input into the CNN, as shown in Fig. \ref{DQN-based system}. The CNN consists of two convolutional (Conv) layers and two fully connected (FC) layers, with parameters chosen to achieve a good performance according to the experiment results as listed in Table \ref {table2}. The filter weights of the four layers in the CNN at time $k$ are denoted by $\theta^{(k)}$ for simplicity. The first Conv layer includes 20 different filters. Each filter has size $2\times2$ and uses stride 1. The output of the first Conv is 20 different $4\times4$ feature maps that are then passed through a rectified linear function (ReLU) as an activation function. The second Conv layer includes 40 different filters. Each filter has size $2\times2$ and stride 1. The outputs of the 2nd Conv layer are 40 different $3\times3$ feature maps, which are flattened to a 360-dimension vector and then sent to the two FC layers. The first FC layer involves 180 rectified linear units, and the second FC layer provides the Q-value for each CPU allocation policy $\textbf{M} \in \triangle_D$ at the current system sequence $\boldsymbol{\varphi}^{(k)}$.
\begin{table}[!hbp]
\caption{CNN Parameters}
\centering
\begin{tabular}{|c|c|c|c|c|}
\hline
Layer  & Conv1 & Conv2  & FC1  & FC2 \\
\hline
Input  &$5\times5$  &$4\times4\times20$  & 360 & 180 \\
\hline
Filter size & $2\times2$ & $2\times2$ & / & / \\
\hline
Stride & 1 & 1& / & /\\
\hline
\# Filters  & 20 & 40 & 180 & $|\triangle_D |$\\
\hline
Activation & ReLU & ReLU & ReLU & ReLU\\
\hline
Output & $4\times4\times20$ & $3\times3\times40$  &180 &$|\triangle_D |$\\
\hline
\end{tabular}\label{table2}
\end{table}

The Q-function as the expected long-term reward for the state sequence $\boldsymbol{\varphi}$ and the action $\textbf{M}$, is given by definition as
\begin{align} \label {loss function}
Q\left(\boldsymbol{\varphi}^{(k)},\textbf{M}\right) = \mathbb{E}_{\boldsymbol{\varphi}'}\left[u_D^{(k)} + \gamma  \mathop{\max}\limits_{\textbf{M}'} Q\left(\boldsymbol{\varphi}',\textbf{M}'\right)|\boldsymbol{\varphi}^{(k)},\textbf{M}\right],
\end{align}
where $\boldsymbol{\varphi}'$ is the next state sequence by choosing defense CPU allocation $\textbf{M}$ at state  $\boldsymbol{\varphi}^{(k)}$.

To make a tradeoff between exploitation and exploration, the defense CPU allocation is chosen according to the $\varepsilon$-greedy policy \cite{rodrigues2009dynamic}. More specifically, the CPU allocation $\textbf{M}^{(k)}$ that maximizes the Q-function is chosen with a high probability $1-\varepsilon$, and other actions are selected with a low probability to avoid staying in the local maximum, i.e.,
\begin{align}
\Pr \left({\textbf{M}^{(k)}} = {\widehat{\textbf{M}}}\right)=
\begin {cases}
1-\varepsilon, & {\widehat{\textbf{M}}}= {\mathop{\rm arg}\nolimits}~{\mathop{\rm max}\nolimits}_{\textbf{M}'} Q\left(\boldsymbol{\varphi}^{(k)},{\textbf{M}'}\right)\\
\frac{\varepsilon}{|\triangle_D |-1},&\textrm{o.w}.
\end{cases} \label {greedy}
\end{align}

Based on the experience replay as shown in Fig. \ref {DQN-based system}, the CPU allocation experience at time $k$ denoted by $\textbf{e}^{(k)}$ is given by $\textbf{e}^{(k)} = \left(\boldsymbol{\varphi}^{(k)}, \textbf{M}^{(k)}, u_D^{(k)}, \boldsymbol{\varphi}^{(k+1)} \right)$, and saved in the replay memory pool denoted by  $\mathcal{D}$, with $\mathcal{D} = \left\{\textbf{e}^{(1)}, \cdots, \textbf{e}^{(k)}\right\}$.
An experience $\textbf{e}^{(d)}$ is chosen from the memory pool at random, with $1 \leq d \leq k$. The CNN parameters $\theta^{(k)}$ are updated by the stochastic gradient descent (SGD) algorithm in which the mean-squared error between the network's output and the target optimal Q-value is minimized with the minibatch updates. The loss function denoted by $L$ in the stochastic gradient descent algorithm is chosen as
\begin{align}
L\left(\theta^{(k)}\right) = \mathbb{E}_{\boldsymbol{\varphi}, \textbf{M}, u_D^{(k)}, \boldsymbol{\varphi}'}\left\{\left(G - Q\left(\boldsymbol{\varphi}, \textbf{M}; \theta^{(k)}\right)\right)^2\right\},
\end{align}
where the target value denoted by $G$ approximates the optimal value $u_D^{(k)} + \gamma Q^*\left(\boldsymbol{\varphi}',\textbf{M}'\right)$ based on the previous CNN parameters $\theta^{(k-1)}$, and is given by
\begin{align} \label{R-target}
G= u_D^{(k)} + \gamma \mathop {\max }\limits_{\textbf{M}'} Q\left(\boldsymbol{\varphi}',\textbf{M}'; \theta^{(k-1)}\right).
\end{align}

\renewcommand{\baselinestretch}{1.0}
\begin{algorithm}[h]
\centering
\caption{Hotbooting process for Algorithm 3}
\begin{algorithmic}[1]
\STATE Initialize $\textbf{N}^{(0)}$, $\textbf{B}^{(1)}$, $\theta^{(0)}$, $\xi$, $K$ and $\mathbb{E} = \emptyset$
\FOR {$i = 1, 2, 3,..., \xi$}
\STATE Emulate a similar CPU allocation scenario for the defender to scan storage devices
\FOR {$k = 1, 2, 3,...,K$}
\STATE Observe the output of the CNN to obtain $Q\left(\boldsymbol{\varphi}^{(k)}, \textbf{M}\right)$
\STATE Choose $\textbf{M}^{(k)}$ via (\ref{greedy})
\FOR {$i = 1, 2,...D$}
\STATE Allocate $M_i^{(k)}$ CPUs to scan storage device $i$
\ENDFOR
\STATE Observe the compromised storage devices and estimate $\textbf{N}^{(k)}$
\STATE Obtain $u_D^{(k)}$ via (\ref{u})
\STATE Observe the resulting state sequence $\boldsymbol{\varphi}^{(k+1)}$
\STATE $\mathbb{E }\leftarrow \mathbb{E} \bigcup \left(\boldsymbol{\varphi}^{(k)}, \textbf{M}^{(k)}, u_D^{(k)}, \boldsymbol{\varphi}^{(k+1)} \right)$
\STATE Perform minibatch update as steps 19-23 in Algorithm 3 to update $\theta^{(k)}$
\ENDFOR
\ENDFOR
\STATE Output $\overline{\theta} \leftarrow \theta^{(k)}$
\end{algorithmic}
\end{algorithm}
The gradient of the loss function with respect to the weights $\theta^{(k)}$ is given by
\begin{align}\label{updateweight} \nonumber
&\nabla_{\theta^{(k)}}L\left(\theta^{(k)}\right) = \mathbb{E}_{\boldsymbol{\varphi}, \textbf{M}, u_D, \boldsymbol{\varphi}'}\left[G \nabla_{\theta^{(k)}}Q\left(\boldsymbol{\varphi}, \textbf{M}; \theta^{(k)}\right)\right] \\
& \ \ \ \ \ \ - \mathbb{E}_{\boldsymbol{\varphi},\textbf{M}}\left[Q\left(\boldsymbol{\varphi}, \textbf{M}; \theta^{(k)}\right) \nabla_{\theta^{(k)}} Q\left(\boldsymbol{\varphi}, \textbf{M}; \theta^{(k)}\right)\right].
\end{align}
This process repeats $H$ times to update $\theta^{(k)}$ in Algorithm 3.

Similar to Algorithm 1, we apply the hotbooting technique to initialize the CNN parameters in the DQN-based CPU allocation rather than initializing them randomly to accelerate the learning speed. As shown in Algorithm 4, the defender stores the emulational experience $\left(\boldsymbol{\varphi}^{(k)}, \textbf{M}^{(k)}, u_D^{(k)}, \boldsymbol{\varphi}^{(k+1)} \right)$ in the database $\mathbb{E}$ and the resulting $\overline{\theta}$ based on $\xi$ experiences are used to set $\theta$ as shown in  Algorithm 3.

\section{Simulation Results}
Simulations have been performed to evaluate the APT defense performance of the CPU allocation schemes in a cloud storage system, with the CNN parameters as listed in Table \ref {table2}. In the simulations, some APT attackers applied the $\varepsilon$-greedy algorithm to choose the number of CPUs to attack each of the $D$ storage devices based on the defense history and some smarter attackers first induced the defender to use a specific ``optimal'' defense strategy based on the estimated defense learning algorithm and then attacked the system accordingly. We set $\alpha = 0.9$, $\gamma = 0.5$, $\delta =0.02$,  $W = 12$, and $H = 16$, if not specified otherwise, to achieve good security performance according to the experiments not presented in this paper.
\begin{figure}[!t]
\begin{center}
\includegraphics[height= 2.5 in]{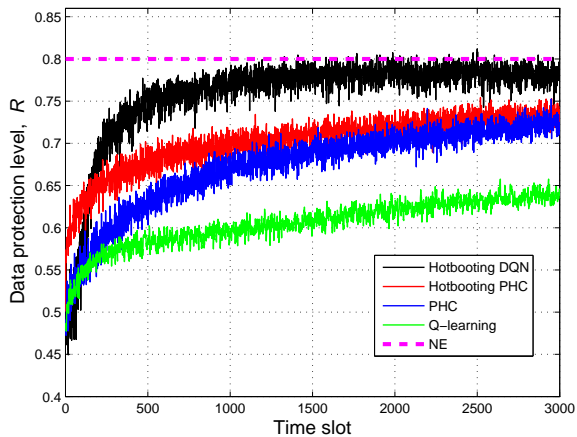} \\
{\footnotesize (a) Data protection level}\\
\includegraphics[height= 2.5 in]{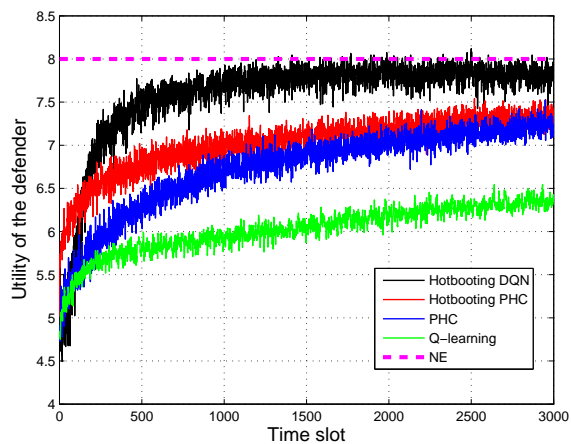} \\
{\footnotesize (b)  Utility of the defender} \\
\caption{APT defense performance of the cloud storage system with 10 storage devices and 10 defense CPUs against an APT attacker with 2 attack CPUs. The size of data stored in each storage device is 1.} \label {LEARN_NE}
\end{center}
\end{figure}

In the first simulation, the defender with $10$ CPUs resisted the attacker with $2$ CPUs over 10 storage devices, each with normalized data size.
As shown in Fig. \ref{LEARN_NE}, the hotbooting DQN-based CPU allocation scheme achieves the optimal policy in a dynamic APT defense game after convergence, which matches the theoretical results of the NE given by Theorem \ref {theorem2}. For example, the data privacy level almost converge to the NE given by (\ref {12}), and the utility of the defender almost converge to the NE given by (\ref {T2U}).
The hotbooting DQN-based CPU allocation scheme outperforms the hotbooting PHC with a faster learning speed, a higher data protection level and a higher utility. The latter in turn exceeds both PHC and Q-learning. For instance, the data protection level of the hotbooting DQN-based scheme is 14.92\% higher than the PHC-based scheme at time slot 1000, which is 30.51\% higher than the Q-learning based scheme. As a result, the hotbooting DQN-based scheme has a 14.92\% higher utility than the PHC-based strategy at time slot 1000, which is 30.51\% higher than that of the Q-learning based strategy. As the hotbooting DQN-based algorithm, an extension of Q learning, compress the learning state space by using CNN to accelerate the learning process and enhance the security performance of the cloud storage system. If the interaction time is long enough, the hotbooting PHC and Q-learning scheme can also converge to the NE of the theoretical results in Theorem \ref {theorem2}.
The PHC-based scheme has less computation complexity than DQN. For example, the PHC-based strategy takes less than 4\% of the time to choose the CPU allocation in a time slot compared with the DQN-based scheme.

\begin{figure}[!t]
\begin{center}
\includegraphics[height= 2.5 in]{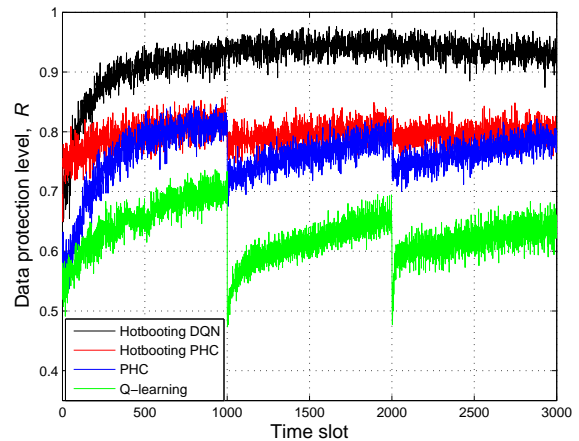} \\
{\footnotesize (a) Data protection level}\\
\includegraphics[height= 2.5 in]{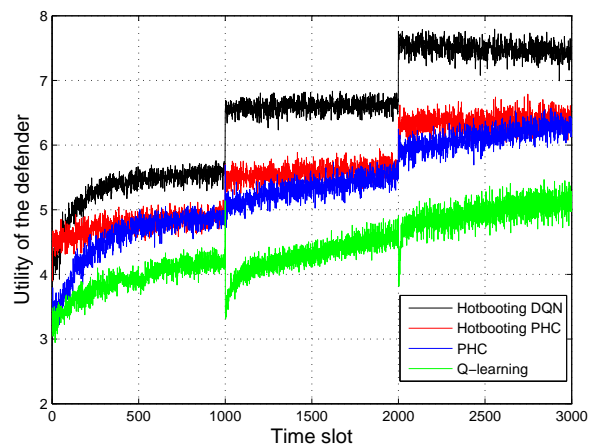} \\
{\footnotesize (b)  Utility of the defender} \\
\caption{APT defense performance of the cloud storage system with 3 storage devices and 16 defense CPUs against an APT attacker with 4 attack CPUs. Both the size of data stored on each device and the attack policy change every 1000 time slots.}\label {BchangeandDirectattack}
\end{center}
\end{figure}

In the second simulation, the size of the data stored in each of the $3$ storage devices of the cloud storage system changed every 1000 time slots. The total data size increases 1.167 times at the 1000-nd time slot and then increases 1.143 times at the 2000-nd time slot. The cloud storage system used $16$ CPUs to scan the storage devices and the APT attacker used $4$ CPUs to attack them. Besides, the attack policy changed every 1000 time slots. The APT attacker estimated the ``optimal'' defense CPU allocation due to the learning algorithm and launched an attack specifically against the estimated defense strategy at time slot 1000 and 2000 to steal data from the cloud storage system. As shown in Fig. \ref {BchangeandDirectattack}, the hotbooting DQN-based CPU allocation is more robust against smart APTs and the time-variant cloud storage system. For example, the data protection level of the hotbooting DQN-based scheme is 30.98\% higher than that of the PHC-based scheme at time slot 1000, which is 97.87\% higher than that of the Q-learning based scheme. As a result, the hotbooting DQN-based scheme has a 30.69\% higher utility than the PHC-based strategy at time slot 1000, which is 96.97\% higher than the Q-learning based strategy.
\begin{figure}[!t]
\begin{center}
\includegraphics[height= 2.5 in]{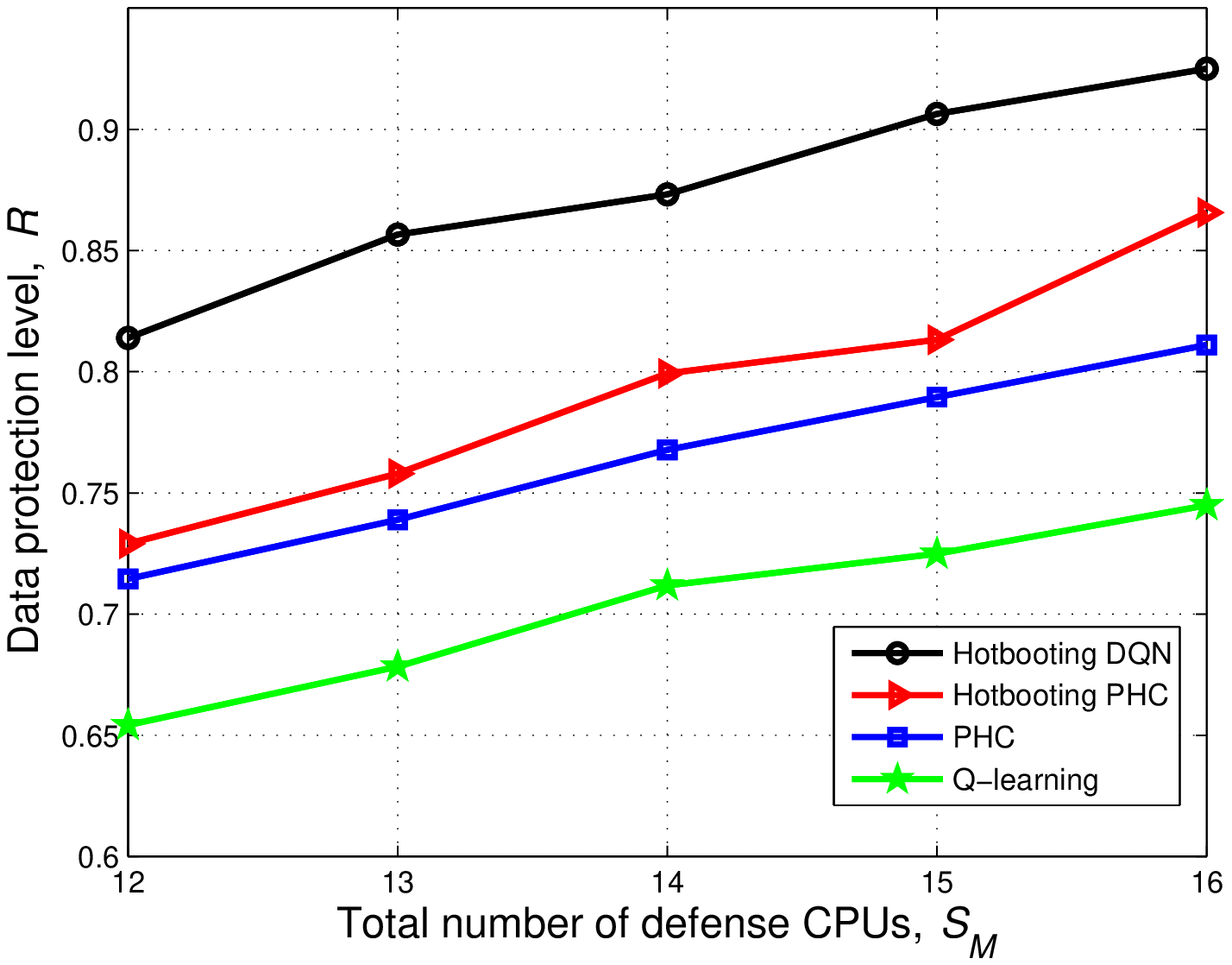} \\
{\footnotesize (a) Data protection level}\\
\includegraphics[height= 2.5 in]{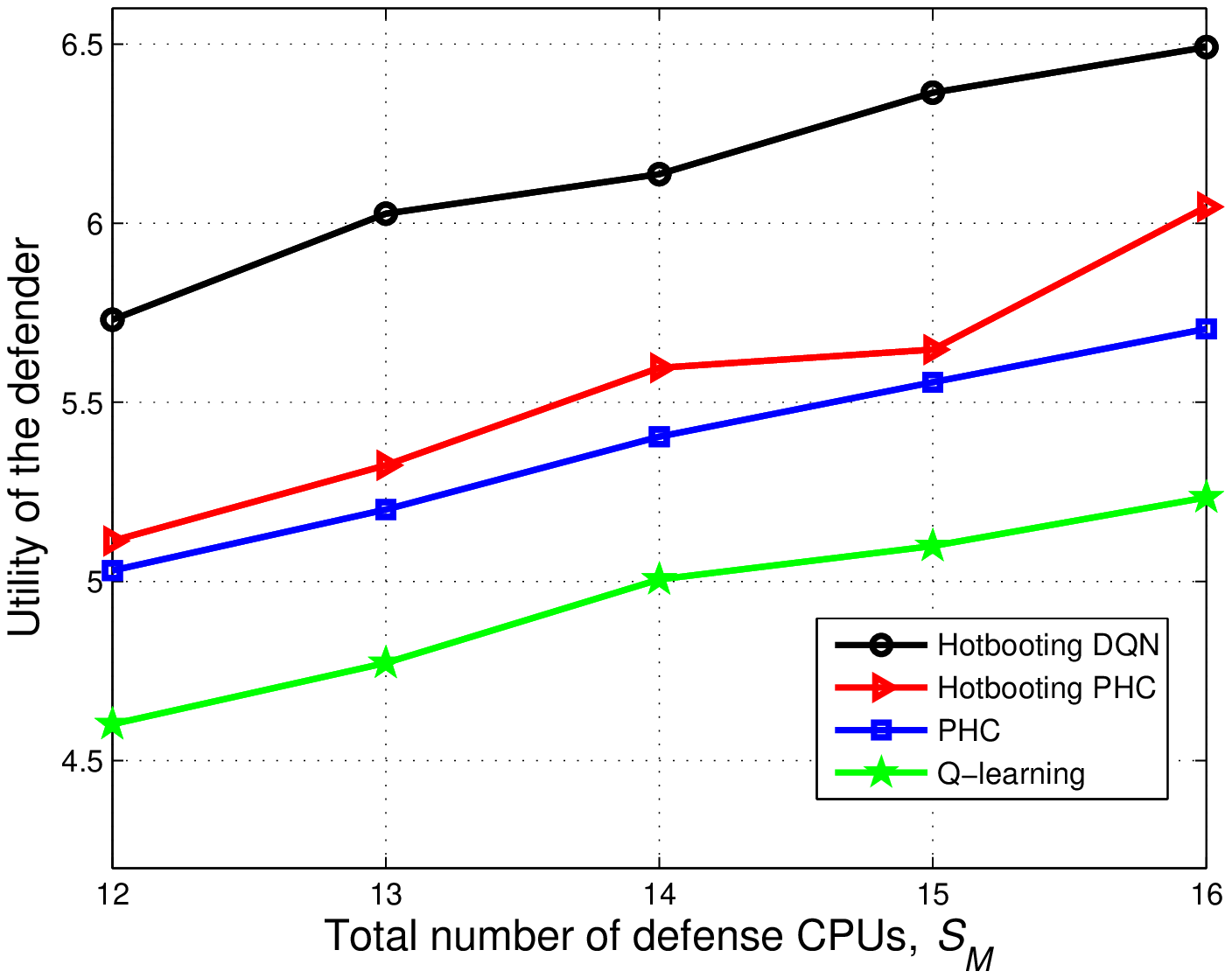} \\
{\footnotesize (b)  Utility of the defender} \\
\caption{APT defense performance of the cloud storage system with $S_M$ defense CPUs, and 3 storage devices that are attacked by 4 attack CPUs, averaged over 3000 time slots. The size of data stored in each storage device changes every 1000 time slots.}\label {AveragewithchangeSM}
\end{center}
\end{figure}

As shown in Fig. \ref {AveragewithchangeSM}, both the data protection level and the utility increase with the number of defense CPUs. For instance, if the number of the defense CPUs changes from 12 to 16, the data protection level and the utility of the defender with hotbooting DQN-based APT defense increase by 14.20\% and 14.03\%, respectively. In the dynamic game with $S_M = 16$, $D = 3$ and $S_N = 4$, the data protection level of the hotbooting DQN-based scheme is 15.85\% higher than that of PHC, which is 21.62\% higher than Q-learning, and the utility of the hotbooting DQN-based CPU allocation scheme is 15.04\% higher than PHC, which is 21.64\% higher than Q-learning.
\begin{figure}[!t]
\begin{center}
\includegraphics[height= 2.5 in]{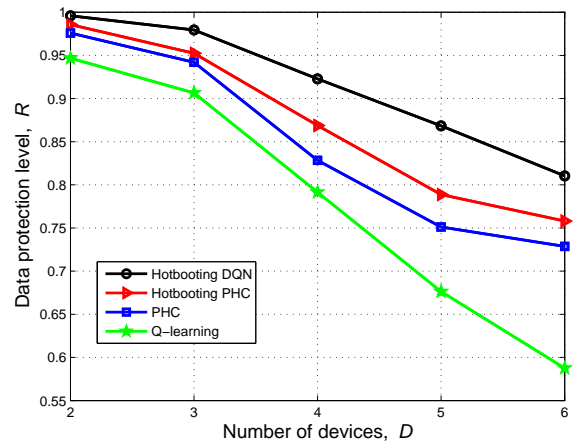} \\
{\footnotesize (a) Data protection level}\\
\includegraphics[height= 2.5 in]{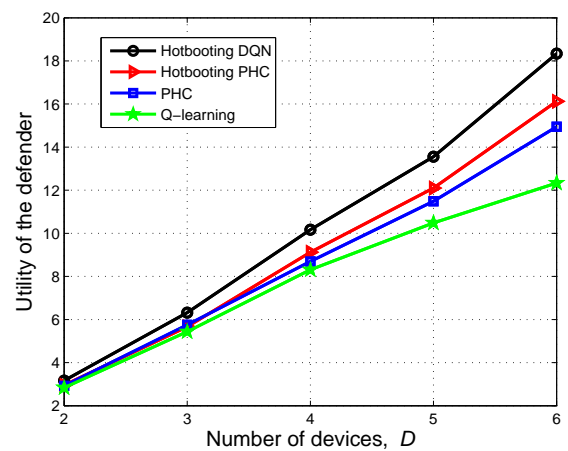} \\
{\footnotesize (b)  Utility of the defender} \\
\caption{APT defense performance of the cloud storage system with $D$ storage devices and 21 defense CPUs against an APT attacker with 4 attack CPUs, averaged over 3000 time slots. The size of data stored in each storage  device changes every 1000 time slots.}\label {AveragewithchangeD}
\end{center}
\end{figure}

As shown in Fig. \ref{AveragewithchangeD}, the APT defense performance slightly decreases with more number of storage devices $D$ in a cloud storage system. However, the hotbooting DQN-based scheme still maintains a high data protection level, i.e., $R=92.50\%$, if 4 storage devices are protected by 21 CPUs and attacked by 4 CPUs. In another example, the hotbooting DQN-based scheme can protect up to 80.05\% of the data stored in 6 storage devices in the cloud. The defender with less number of CPUs has to distribute its resources among all the storage devices to resist APTs, as at least one CPU has to scan each storage device. The performance gain of the hotbooting DQN-based CPU allocation scheme over the hotbooting PHC-based scheme increases with the number of storage devices in the cloud system.

\section{Conclusion}
In this paper, we have formulated a CBG-based CPU allocation game for the APT defense of cloud storage and cyber systems and provided the NEs of the game to show how the number of storage devices, the data sizes in the storage devices and the total number of CPUs impact on the data protection level of the cloud storage system and the defender's utility.
A hotbooting DQN-based CPU allocation strategy has been proposed for the defender to scan the storage devices without being aware of the attack model and the data storage model in the dynamic game. The proposed scheme can improve the data protection level with a faster learning speed and is more robust against smart APT attackers that choose the attack policy based on the estimated defense learning scheme. For instance, the data protection level of the cloud storage system and the utility of the defender increases by 22.29\% and 22.4\%, respectively, compared with the Q-learning based scheme in the cloud storage system with 4 storage devices and 16 defense CPUs against an APT attacker with 4 CPUs. A hotbooting PHC-based CPU allocation scheme can reduce the computation complexity.
\begin{spacing}{1.0}
\bibliography{citeCBG}
\bibliographystyle{ieeetr}
\end{spacing}

\end{document}